\documentclass[12pt]{article}

\usepackage{csquotes}
\usepackage{amsmath}
\usepackage{amsthm}
\usepackage{amssymb}
\usepackage{enumerate}
\usepackage{natbib}
\usepackage{url}

\usepackage{algorithm}
\usepackage{algpseudocode}

\newtheorem{lemma}{Lemma}

\newtheorem{theorem}{Theorem}

\usepackage{hyperref}

\newcommand{\R}{\mathbb{R}}
\newcommand{\N}{\mathbb{N}}

\newcommand{\E}{\mathbb{E}}
\newcommand{\bP}{\mathbb{P}}

\newcommand{\F}{\mathcal{F}}

\newcommand{\iid}{\overset{\mathrm{iid}}{\sim}}

\begin{document}

\markboth{T. Hilbert}{Predictive Inference via Kernel Density Estimates}

\title{Predictive Inference via Kernel Density Estimates}

\author{Torey Hilbert
\hspace{.2cm}\\
    Department of Statistics, Ohio State University
}

\maketitle

\begin{abstract}
Kernel density estimation is a widely used nonparametric approach to estimate an unknown distribution. Recent work in Bayesian predictive inference has considered stochastic processes formed by specifying the predictive distribution for the next data point given all observed data such that the resulting predictive distributions converge weakly almost surely. We study two kernel based prediction rules: the classic kernel density estimator, and a recursive version previously introduced for online problems. We show that both processes converge weakly almost surely, which opens the door for new Bayesian interpretations of kernel density estimation. Surprisingly, the process based on the classic kernel density estimates converges to a compactly supported measure, while the recursive version converges to a non-compactly supported measure.
\end{abstract}

\section{Introduction}\label{sec-intro}

Recent work on Bayesian predictive inference has considered modeling data by specifying a collection of prediction rules for $X_{n+1} | X_{1:n}$, which when supplemented with a distribution over $X_1$ gives a joint distribution over an infinite sequence of data $X_{1:\infty}$.
Denoting the (random) predictive measures by $P_n(B) = \bP(X_{n+1} \in B | X_{1:n})$ for Borel sets $B \in \mathcal{B}(\R^d)$, if the sequence of predictive measures $P_n$ almost surely converges weakly to some (random) probability measure $P$, then the conditional limiting predictive measure $P | X_{1:n}$ is a posterior distribution for a Bayesian analysis
\citep{fortini2020quasi, fong2023martingale, fortini2025exchangeability}.
The recent Bayesian predictive inference literature has largely focused on predictive distributions satisfying the conditionally identically distributed (CID) property \citep{berti2004limit, fong2023martingale, fortini2025exchangeability}, or on predictive distributions that are ``almost'' a CID sequence as in \citet{battiston2025bayesian}.

Part of the appeal of the Bayesian predictive inference framework, however, is the idea that the analyst could use any predictive distributions that they feel are ``natural''.
When handed a non-parametric prediction problem, the first approach for many statisticians is to use the empirical distribution, leading to Bayesian bootstrap  methods \citep{ferguson1973bayesian, rubin1981bayesian, ghosh1997bayesian}. These methods are generalized in the works on measure-valued P\'olya sequences \citep{berti2023kernel, sariev2023infinite, chorbadzhiyska2025mvps}. A second natural approach for many statisticians is to smooth the empirical distribution, for example with a kernel density estimator \citep{silverman2018density}.

Let $K$ be a probability measure on $\R^d$, and let $h_n > 0$ be a sequence of bandwidths. For each $n$ and Borel set $B \in \mathcal{B}(\R^d)$, we consider the predictive rule
\begin{align}\label{eq-kernel}
    \bP(X_{n+1} \in B | X_{1:n}) = \frac{1}{n}\sum_{i = 1}^n K\bigg(\frac{B - X_i}{h_n}\bigg).
\end{align}
Here we use the notation $(B-X_i)/{h_n} = \{ (b - X_i) / h_n : b\in B\} \subseteq \R^d$.
To fully specify the model we will arbitrarily set $X_1 = 0$. Denote the sequence of random predictive measures by $P_n(B) = \bP(X_{n+1} \in B | X_{1:n})$.
For example, if $K = N(0, 1)$ on $\R$, this is equivalent to using the probability density function
\begin{align*}
    f(x | X_{1:n}) = \frac{1}{n} \sum_{i = 1}^n \frac{1}{h_n} \frac{1}{\sqrt{2\pi}} \exp \bigg(-\frac{1}{2}\bigg(\frac{x - X_i}{h_n}\bigg)^2\bigg),
\end{align*}
the classic kernel density estimator (KDE).

In this work, we show that for modestly heavy-tailed kernels $K$ and bandwidths $h_n$ that decay at order $n^{-\delta}$ for some $\delta > 0$, the sequence of predictive distributions defined by the KDE process above converges weakly almost surely to a limiting random measure $P$. Furthermore, with additional moment assumptions, the limiting random measure $P$ is supported on a random compact set. While the almost sure weak convergence of $P_n$ is expected, the fact that the limiting measure $P$ is almost surely compactly supported was surprising to us.

While our primary interest is model~\ref{eq-kernel}, we also consider a recursive kernel estimator introduced by \citet{wolverton1969asymptotically}.
For standard implementations, the KDE process and this alternative recursive kernel process behave qualitatively differently, since the limiting measure of the recursive kernel process almost surely has non-compact support. For the recursive kernel estimator, \citet{battiston2025bayesian} find that for very rapidly decaying bandwidths, such as $h_n = e^{-n}$, the sequence of predictive measures converges weakly almost surely for some kernels $K$. However, their approach does not work for $h_n$ satisfying $n h_n \to \infty$, which is a well-known condition for frequentist point-wise consistency of the KDE \citep{silverman2018density}. We show, under much weaker conditions, that the recursive kernel estimator also converges weakly almost surely.

\section{Kernel Density Process}\label{sec-main}

Another way to write the same sequence of variables in model~\ref{eq-kernel} is as follows:
\begin{align*}
    Y_{1:\infty} &\iid K,\\
    M_n &\overset{ind}{\sim} \text{Unif}(\{1, \ldots, n\}), \\
    X_{n + 1} &= X_{M_n} + h_n Y_n.
\end{align*}
In contrast to $X_{1:\infty}$, we also consider a different process $\tilde{X}_{1:\infty}$ defined by recursive kernel estimates, as in \citet{wolverton1969asymptotically, battiston2025bayesian}. We will use a tilde to differentiate objects connected to $X_{1:\infty}$ vs $\tilde{X}_{1:\infty}$. We define the recursive kernel process by
\begin{align}\label{eq-recursive}
    \tilde{P}_n(B) = \bP(\tilde{X}_{n + 1} \in B | \tilde{X}_{1:n}) = \bigg(1 - \frac{1}{n}\bigg)\tilde{P}_{n - 1}(B) + \frac{1}{n} K\bigg(\frac{B - \tilde{X}_n}{h_n}\bigg).
\end{align}
Again we arbitrarily select $\tilde{X}_1 = 0$ to start the sequence. The key difference between $X_{1:\infty}$ and $\tilde{X}_{1:\infty}$ is that the former model uses the same bandwidth $h_n$ around all points and recomputes the KDE at every step, while the latter model never updates the bandwidth for an already observed point in the sequence. Subjectively, we prefer the classic KDE approach defining the sequence $X_{1:\infty}$, but, as the two models have qualitatively different asymptotic behavior, we analyze both. Another way to write the recursive process is
\begin{align*}
    Y_{1:\infty} &\iid K,\\
    M_n &\overset{ind}{\sim} \text{Unif}(\{1, \ldots, n\}), \\
    \tilde{X}_{n + 1} &= \tilde{X}_{M_n} + h_{M_n} Y_n.
\end{align*}

The following two results show that both the classic KDE predictive rule and the recursive predictive rule are asymptotically exchangeable, and thus can be used in the Bayesian predictive inference paradigm laid out by \citet{fortini2025exchangeability}.

\begin{theorem}[Weak convergence for KDE process]\label{thm-kde-conv}
    Suppose $0 < h_n \leq Cn^{-\delta}$ for some $C, \delta > 0$, and suppose $\E(\|Y_1\|) < \infty$. Then $\bP(P_n \overset{weak}{\to} P) = 1$ for some random probability measure $P$ on $\R^d$. Additionally, if $\E(\|Y_1\|^{1/c}) < \infty$ for some $0 < c < \delta$, then $P$ is almost surely supported on a (random) compact set $V$, and $\bP(\sup_{n \in \N} \|X_n\| < \infty) = 1$.
\end{theorem}

\begin{theorem}[Weak convergence for recursive process]\label{thm-recursive-conv}
    Suppose $0 < h_n \leq Cn^{-\delta}$ for some $C, \delta > 0$, and suppose $\E(\|Y_1\|) < \infty$. Then $\bP(\tilde{P_n} \overset{weak}{\to} \tilde{P}) = 1$ for some random probability measure $\tilde{P}$ on $\R^d$.
\end{theorem}

A common choice for densities on $\R$ is $h_n = O(n^{-1/5})$, which in a frequentist setting minimizes asymptotic mean integrated square error under regularity conditions e.g., \citet{jones1996brief, silverman2018density}.
Our theorems establish that, for any choice of kernel $K$ with finite first moment, both the KDE process $P_n$ and the recursive process $\tilde{P}_n$ will converge weakly almost surely.
Furthermore, if $K$ has finite sixth moments then $P$ almost surely has compact support.
For densities on $\R^d$, typically we use $h_n = O(n^{-1/(d+4)})$, so that if $\|Y_1\|$ has $d + 4$ moments then $P$ has compact support. Lastly, when $K$ has exponential tails, $P$ has compact support for any $\delta > 0$.

Typically one would use a kernel $K$ with $\E[Y_1] = 0$, but to demonstrate a critical difference between the KDE process and the recursive kernel process, we will consider an artificial example with $K([0, \infty)) = 1$ that is substantially more straightforward mathematically.
Let $K = |N(0, 1)|$, the absolute value of a $N(0, 1)$ random variable, and set $h_n = n^{-\delta}$.
Immediately, we see that $\tilde{P}_n$ will always include a $\frac{1}{n}|N(0, h_1^2)|$ term coming from $X_1 = 0$, so by the Borel-Cantelli lemma with $\sum_{n \in \N} 1/n = \infty$, we have $\bP(\sup_{n \in \N} \|\tilde{X}_n\| = \infty) = 1$, in contrast to $\bP(\sup_{n \in \N} \|X_n\| < \infty) = 1$.
Fix $M$ and suppose $\tilde{X}_n > M$ for some $n$. Let $\tilde{p}_m = \frac{1}{m}|\{j \geq n: \tilde{X}_j \text{ is a descendant of }\tilde{X}_n\}|$. Using a Polya urn scheme, we get $\tilde{p}_m \overset{a.s.}{\to} \tilde{p} \sim \text{Beta}(1, n-1)$. Then
\begin{align*}
    \tilde{P}_m(\{ x: |x| > M\}) \geq \tilde{p}_m \not\to 0.
\end{align*}
But for any $M > 0$, almost surely $\tilde{X}_n > M$ for some $n$. Thus the limiting measure $\tilde{P}$ is not compactly supported. The same result holds for $K = N(0, 1)$, albeit with more tedious details since the descendants of $\tilde{X}_n$ are no longer strictly greater than $\tilde{X}_n$.

\section{Proofs}\label{sec-proof}

\subsection{Notation}
First we consider the non-negative processes on $\R$ given by defining $W_n = \|Y_n\|$ and
\begin{align*}
    U_{n + 1} = U_{M_n} + h_{n} W_n, \quad \tilde{U}_{n + 1} = \tilde{U}_{M_n} + h_{M_n} W_n,
\end{align*}
where $U_1 = 0$ and $\tilde{U}_1 = 0$.
Notice that $\|X_n\| \leq U_n$ and $\|\tilde{X}_n\| \leq \tilde{U}_n$. Throughout we will let $\F_n = \sigma(\{X_k, \tilde{X}_k : k \leq n\})$ be the sigma-field generated by data up to step $n$.

\subsection{Tightness for the KDE process}

Define $J_n = \frac{1}{n}\sum_{i = 1}^n U_i$. Then
\begin{align*}
    \E[U_{n+1}| \F_n] &= \E[U_{M_n} + h_n W_n | \F_n] = \frac{1}{n}\sum_{i = 1}^n U_i + h_n \E[W_1],\\
    \E[J_{n+1} | \F_n] &= \frac{n}{n+1} J_n + \frac{1}{n+1} \E[U_{n + 1} | \F_n]  = J_n + c_n,
\end{align*}
where $c_n = \frac{1}{n+1} h_n \E[W_1]$. Notice that $\sum_{k = 1}^\infty c_k < \infty$. Defining the sequence of non-negative variables $S_n = J_n + \sum_{k = n}^\infty c_k$, we get
\begin{align*}
    \E[S_{n+1} | \F_n] = J_n + c_n + \sum_{k = n + 1}^\infty c_k = S_n,
\end{align*}
so that $S_n$ is a non-negative martingale. Then $S_n \overset{a.s.}{\to} S$, and $\sum_{k = n}^\infty c_k \to 0$, so we get the almost sure convergence $J_n \overset{a.s.}{\to} S$.
Fix $\epsilon > 0$, and let $N$ be such that for all $n \geq N$, $|J_n - S| < \epsilon$ and $h_n \E[W_1] < \epsilon$. Defining the predictive measures $Q_n(B) = \bP(U_{n+1} \in B | \F_n)$, 
\begin{align*}
    Q_n((S/\epsilon + 1, \infty))  &\leq \frac{J_n + h_n \E[W_1]}{S/\epsilon + 1} \leq \epsilon\Big(\frac{S + \epsilon}{S + \epsilon} + \frac{\epsilon}{S + \epsilon}\Big) \leq 2\epsilon.
\end{align*}
But since $\|X_n\| \leq U_n$, this implies $\limsup_n P_n(\{ x: \|x\| > S/\epsilon + 1\}) \leq 2\epsilon$, showing that $\{P_n\}_{n \in \N}$ is almost surely a tight collection of measures.

\subsection{Tightness for the recursive process}

Similarly, we define $\tilde{J}_n = \frac{1}{n}\sum_{i = 1}^n \tilde{U}_i$. Then
\begin{align*}
    \E[\tilde{U}_{n+1}| \F_n] &= \E[\tilde{U}_{M_n} + h_{M_n} W_n | \F_n] = \frac{1}{n}\sum_{i = 1}^n \tilde{U}_i + \frac{1}{n}  \sum_{i = 1}^n h_i \E[W_1],\\
    \E[\tilde{J}_{n+1} | \F_n] &= \frac{n}{n+1} \tilde{J}_n + \frac{1}{n+1} \E[\tilde{U}_{n+1} | \F_n] = \tilde{J}_n + c_n,
\end{align*}
where $c_n = \frac{1}{n(n+1)} \sum_{i = 1}^n h_i \E[W_1]$. Notice that $\sum_{k = 1}^\infty c_k < \infty$. Defining the sequence of non-negative variables $\tilde{S}_n = \tilde{J}_n + \sum_{k = n}^\infty c_k$, we get
\begin{align*}
    \E[\tilde{S}_{n+1} | \F_n] = \tilde{J}_n + c_n + \sum_{k = n + 1}^\infty c_k = \tilde{S}_n,
\end{align*}
so that $\tilde{S}_n$ is a non-negative martingale. Then $\tilde{S}_n \overset{a.s.}{\to} \tilde{S}$, and $\sum_{k = n}^\infty c_k \to 0$, so we get the almost sure convergence $\tilde{J}_n \overset{a.s.}{\to} \tilde{S}$.
Fix $\epsilon > 0$, and let $N$ be such that for all $n \geq N$, $|\tilde{J}_n - \tilde{S}| < \epsilon$ and $\frac{1}{n}\sum_{k = 1}^n h_k \E[W_1] < \epsilon$. Defining the predictive measures $\tilde{Q}_n(B) = \bP(\tilde{U}_{n+1} \in B | \F_n)$, 
\begin{align*}
    \tilde{Q}_n((\tilde{S}/\epsilon + 1, \infty))  &\leq \frac{\tilde{J}_n + \frac{1}{n}\sum_{k = 1}^n h_k \E[W_1]}{\tilde{S}/\epsilon + 1} \leq \epsilon\Big(\frac{\tilde{S} + \epsilon}{\tilde{S} + \epsilon} + \frac{\epsilon}{\tilde{S} + \epsilon}\Big) \leq 2\epsilon.
\end{align*}
But since $\|\tilde{X}_n\| \leq \tilde{U}_n$, this implies $\limsup_n \tilde{P}_n(\{ x: \|x\| > \tilde{S}/\epsilon + 1\}) \leq 2\epsilon$, showing that $\{\tilde{P}_n\}_{n \in \N}$ is almost surely a tight collection of measures.

\subsection{Convergence of KDE characteristic functions}

Let $\phi_n(t)$ be the characteristic function of $P_n$, that is $\E[\exp(it^TX_{n+1}) | \F_n]$. Let $\phi_K(t)$ be the characteristic function of $K$. Then
\begin{align*}
    \phi_n(t) = \frac{1}{n} \sum_{k = 1}^n \exp(it^T X_k) \phi_K(h_n t),
\end{align*}
so that
\begin{align*}
    \E[\phi_{n+1}(t) | \F_n] &= \frac{1}{n+1} \sum_{k = 1}^n \exp(it^T X_k) \phi_K(h_{n+1} t) + \frac{1}{n+ 1} \E[\exp(it^T X_{n+1} | \F_n)]\phi_K(h_{n+1}t) \\
    &= \frac{n}{n+1} \frac{\phi_K(h_{n+1} t)}{\phi_K(h_n t)} \phi_n(t) + \frac{1}{n+1} \phi_n(t) \phi_K(h_{n+1} t)\\
    &= \bigg(\frac{\phi_K(h_n t)}{n+1}  + \frac{n}{n+1}\bigg) \frac{\phi_K(h_{n+1}t)}{\phi_K(h_n t)} \phi_n(t).
\end{align*}
Notice, there exists $N$ such that for all $n \geq N$, $\phi_K(h_n t) \neq 0$ by the continuity of $\phi_K$ at $t = 0$. Define 
\begin{align*}
    a_n(t) &= \frac{\phi_K(h_n t)}{n+1}  + \frac{n}{n+1},\\
    b_n(t) &= a_n(t) \frac{\phi_K(h_{n+1}t)}{\phi_K(h_n t)}, \\
    c_n(t) &= \Big(\prod_{k = n}^{\infty} b_k(t)\Big) = \Big(\prod_{k = n}^{\infty} a_k(t)\Big) \frac{1}{\phi_K(h_n t)},
\end{align*}
where in Lemma~\ref{lem-product} we show that $\prod_{k = n}^{\infty} a_k(t) \neq  0$. Then $\E[\phi_{n+1}(t)|\F_n] = b_n(t) \phi_n(t)$.
Define the martingale $S_n(t) = c_n(t) \phi_n(t)$, starting at step $N$ to avoid dividing by $0$, by
\begin{align*}
    \E[S_{n + 1}(t) | \F_n] &= \Big(\prod_{k = n+1}^\infty b_k(t)\Big) \E[\phi_{n+1}(t) |\F_n] = \Big(\prod_{k = n+1}^\infty b_k(t)\Big) b_n(t)\phi_n(t) = S_n(t).
\end{align*}
Since $\sup_{n \geq N}|c_n| < \infty$, we have $|S_n(t)| \leq \sup_{n \geq N}|c_n| < \infty$, so $S_n(t)$ is a bounded martingale. Thus $S_n(t) \overset{a.s.}{\to} S(t)$, so that $\phi_n(t) \overset{a.s.} \to S(t)$ because $c_n \neq 0$ for all $n \geq N$ and $\lim_n c_n(t) = 1$.

We have already shown $\{P_n\}_{n \in \N}$ almost surely is tight, and now the sequence of associated characteristic functions converges almost surely to $S(t)$; thus $P_n \overset{weak}{\to} P$ almost surely for some (random) measure $P$.

\subsection{Convergence of recursive process characteristic functions}

Let $\tilde{\phi}_n(t)$ be the characteristic function of $\tilde{P}_n$, that is $\E[\exp(it^T \tilde{X}_{n+1}) | \F_n]$. Then
\begin{align*}
    \tilde{\phi}_n(t) = \frac{1}{n} \sum_{k = 1}^n \exp(it^T \tilde{X}_k) \phi_K(h_k t),
\end{align*}
so that
\begin{align*}
    \E[\tilde{\phi}_{n+1}(t) | \F_n] &= \frac{1}{n+1} \sum_{k = 1}^n \exp(it^T \tilde{X}_k) \phi_K(h_k t) + \frac{1}{n+ 1} \E[\exp(it^T \tilde{X}_{n+1} | \F_n)]\phi_K(h_{n+1}t) \\
    &= \frac{n}{n+1} \phi_n(t) + \frac{1}{n+1} \tilde{\phi}_n(t) \phi_K(h_{n+1} t).
\end{align*}
Notice that there exists $N$ such that for all $n \geq N$, $\phi_K(h_n t) \neq 0$ by the continuity of $\phi_K$ at $t = 0$. Define
\begin{align*}
    \tilde{a}_n(t) &= \frac{\phi_K(h_{n+1} t)}{n+1}  + \frac{n}{n+1},\\
    \tilde{c}_n(t) &= \prod_{k = n}^{\infty} \tilde{a}_k(t),
\end{align*}
where in Lemma~\ref{lem-product} we show that $\prod_{k = n}^{\infty} a_k(t) \neq  0$. Then $\E[\tilde{\phi}_{n+1}(t)|\F_n] = \tilde{a}_n(t) \tilde{\phi}_n(t)$.
Define the martingale $\tilde{S}_n(t) = \tilde{c}_n(t) \tilde{\phi}_n(t)$, starting at step $N$ to avoid dividing by $0$, by
\begin{align*}
    \E[\tilde{S}_{n + 1}(t) | \F_n] &= \Bigg(\prod_{k = n+1}^\infty \tilde{a}_k(t)\Bigg) \E[\tilde{\phi}_{n+1}(t) |\F_n] = \Bigg(\prod_{k = n+1}^\infty \tilde{a}_k(t)\Bigg) \tilde{a}_n(t) \tilde{\phi}_n(t) = \tilde{S}_n(t).
\end{align*}
Since $|\tilde{c}_n| \leq 1$, we have $|\tilde{S}_n(t)| \leq 1$, so $\tilde{S}_n(t)$ is a bounded martingale. Thus $\tilde{S}_n(t) \overset{a.s.}{\to} \tilde{S}(t)$, so that $\tilde{\phi}_n(t) \overset{a.s.} \to \tilde{S}(t)$ because $\tilde{c}_n \neq 0$ for all $n \geq N$ and $\lim_n \tilde{c}_n(t) = 1$.

We have already shown $\{\tilde{P}_n\}_{n \in \N}$ almost surely is tight, and now the sequence of associated characteristic functions converges almost surely to $\tilde{S}(t)$; thus $\tilde{P}_n \overset{weak}{\to} \tilde{P}$ almost surely for some (random) measure $\tilde{P}$.

\subsection{Compact support for limiting measure of KDE process}

Assume that $\E(\|Y_1\|^{1/c}) < \infty$ for some $0 < c < \delta$.
Let $L_{n, j}$ be the number of descendants of $X_j$ during the steps $\{n+1, n+2, \ldots, 2n\}$; for example, at step $n + 1$, we might have $M_{n + 1} = j$, and then at step $n + 2$ we might have $M_{n + 2} = n+1$, so that there are two descendants of $X_j$, thus $L_{n, j} = 2$. Then $L_{n, j}$ can be modeled using a P\'{o}lya urn, namely as an urn that starts with 1 black ball and $n - 1$ red balls, and $L_{n, j}$ is the number of black balls drawn after $n$ steps.
Then using the beta-binomial form for the P\'{o}lya urn, we have
\begin{align*}
    \bP(L_{n, j} = k) &= \binom{n}{k}\frac{B(k + 1, 2n - k - 1)}{B(1, n-1)} \\
    &= (n - 1)\frac{\prod_{i = 0}^{k - 1} (n - i)}{\prod_{i = 0}^{k - 1}(2n - 1 - i)} \\
    &\leq (n-1)(2/3)^k,\\
    \bP(L_{n, j} \geq k) &\leq \sum_{r = k}^n (n-1)(2/3)^{r} \leq 3(n-1)(2/3)^{k}.
\end{align*}
Then calling $L_n = \max_{j \in \{1, \ldots, n\}} L_{n, j}$,
\begin{align*}
    \bP(L_n \geq k) \leq 3 n(n -1)(2/3)^{k}.
\end{align*}
Now let us partition $\N$ into sets $J_m = \{ j : 2^m + 1\leq j \leq 2^{m + 1}\}$, starting with $J_0$. Pick some $r \in (c, \delta)$. We consider the events
\begin{align*}
    A_m = \{ L_{2^m} \geq 10m\}, \quad B_m = \Big\{ \max_{j \in J_{m}} \|Y_j\| \geq 2^{r m} \Big\}.
\end{align*}
We aim to show that $A_m$ and $B_m$ happen only finitely many times.
\begin{align*}
    \bP(A_m) &\leq 3 \cdot 2^m (2^m - 1) (2/3)^{10m} \leq 3 \cdot \exp(-m (10 \log (3/2) - \log 4)),
\end{align*}
and thus $\sum_{m = 1}^\infty \bP(A_m) < \infty$. Next, we have
\begin{align*}
    \E\Big[\max_{j \in J_m} \|Y_j\| \Big] &\leq \bigg(\E\sum_{j \in J_m}^n \|Y_j\|^{1/c}\bigg)^c \leq 2^{cm} \E(|Y_i|^{1/c})^c,\\
    \bP\Big(\max_{j \in J_m} \|Y_j\| \geq 2^{rm}\Big) &\leq \frac{\E\Big[\max_{j \in J_m} \|Y_j\| \Big]}{2^{rm}} \leq C_1 2^{-m(r-c)},
\end{align*}
so that $\sum_{m = 1}^\infty \bP(B_m) < \infty$.
Thus there exists a random variable $N$ such that for all $m \geq N$, neither $A_m$ nor $B_m$ occur. Then by tracing the genealogy of $X_n$ through each set of steps given by $J_m$,
\begin{align*}
    \|X_n\| &\leq \sum_{m = 0}^{\infty} L_{2^m} h_{2^m} \max_{j \in J_{m}} \|Y_j\| \\
    &\leq \sum_{m = 0}^N L_{2^m} h_{2^m} \max_{j \in J_{m}} \|Y_j\| + \sum_{m = N}^\infty h_{2^m} 10m \cdot 2^{rm} \\
    &\leq \sum_{m = 0}^N L_{2^m} h_{2^m} \max_{j \in J_{m}} \|Y_j\|+ \sum_{m = N}^\infty 10mC (2^m)^{-(\delta - r)} = T < \infty,
\end{align*}
where we have defined the random variable $T$ in the last line. Then $\bP(\sup_{n \in \N} \|X_n\| < \infty) = \bP(\sup_{n \in \N} \|X_n\| \leq T) = 1$. If $\sup_{n \in \N} \|X_n\| < \infty$ then
\begin{align*}
    P_m(\{x : \|x\| > \sup_n \|X_n\| + \epsilon\}) \leq \bP(h_m \|Y_m\| > \epsilon) \to 0.
\end{align*}

\section{Discussion}\label{sec-conc}

Kernel density estimation is an appealing approach in part due to its simplicity, and thus substantial work has gone into understanding whether it can be interpreted in a Bayesian sense, for example \citet{west1991kernel}. Similar to the analysis of Newton's algorithm by \citet{fortini2020quasi}, we can loosely interpret the random limiting measures $P$ and $\tilde{P}$ as draws from a prior distribution over measures on $\R^d$ associated with the predictive rule.
When $K = N(0, 1)$ and $h_n = Cn^{-1/5}$, the support of $P$ being almost surely compact suggests that using the classic KDE approach as a predictive rule loosely assumes the existence of a compact set $V \subseteq \R^d$ that supports the distribution, although $V$ may be very large.
This holds even when using a heavier-tailed kernel. For example, Theorem~\ref{thm-kde-conv} tells us that a $t_7$ distribution for $K$ will give rise to a $P$ supported on a compact set. This is surprising, but can be viewed as a slightly stronger parallel to the well-known fact that the tails of a Dirichlet process mixture model are lighter than the tails of the base measure \citep{ghosal2017fundamentals}.

\section*{Declaration of the use of generative AI and AI-assisted technologies}

During the preparation of this work the author used Claude Opus 4.7 in order to assist in a final literature review. After using this tool/service the author reviewed and edited the content as necessary and takes full responsibility for the content of the publication.

\section*{Acknowledgement}

The author thanks Steven N. MacEachern for helpful feedback on a draft of the paper.

\appendix

\section*{Appendix 1}

\begin{lemma}\label{lem-product}
Suppose $h_n \leq Cn^{-\delta}$ and $W \sim K$ has $\E[\|W\|] < \infty$. There exists $N$ such that for all $n \geq N$,
\begin{align}
    \prod_{k = n}^\infty \bigg(\frac{\phi_K(h_k t)}{k+1} + \frac{k}{k+1}\bigg) &\neq 0. \label{eq-nonzero}
\end{align}
\end{lemma}

\begin{proof}
Fix $t \in \R^d$. Taylor expanding $\phi_K$, we get $\phi_K(h_k t) = 1 + i h_k t^T\E[W] + R_k$, where $|R_k| \leq 2 \big(h_k\|t\|\big)\E[\|W\|]$. Then
\begin{align*}
    \sum_{k = n}^\infty \bigg|\frac{\phi_K(h_k t)}{k+1} + \frac{k}{k+1} - 1\bigg| &= \sum_{k = n}^\infty \frac{1}{k+1} \Big| ih_kt^T \E[W] + R_k\Big| \\
    &\leq \sum_{k = n}^\infty \frac{h_k \|t\|}{k+1}\Big(\|\E[W]\| + 2 \E[\|W\|]\Big) \\
    &\leq \sum_{k = n}^\infty C_2 k^{-(1+\delta)} < \infty.
\end{align*}
Using the continuity of $\phi_K$, for $k \geq n \geq N$ we have $\phi_K(h_{k+1} t) \neq -1$, and the triangle inequality gives us $\frac{\phi_K(h_{k + 1} t)}{k+1} + \frac{k}{k+1}\neq 0$, so the convergence of the above sum implies the convergence of \eqref{eq-nonzero} to a nonzero value.
\end{proof}

\bibliographystyle{abbrvnat}
\bibliography{paper-ref}

\end{document}